\documentclass{article}

\usepackage{arxiv}

\usepackage[utf8]{inputenc} % allow utf-8 input
\usepackage[T1]{fontenc}    % use 8-bit T1 fonts
\usepackage{hyperref}       % hyperlinks
\usepackage{url}            % simple URL typesetting
\usepackage{booktabs}       % professional-quality tables
\usepackage{amsfonts}       % blackboard math symbols
\usepackage{nicefrac}       % compact symbols for 1/2, etc.
\usepackage{microtype}      % microtypography
\usepackage{lipsum}
\usepackage{graphicx}
\usepackage{wrapfig}
\usepackage{amsthm}
\usepackage[ruled,linesnumbered]{algorithm2e}

\newcommand{\RR}{\ensuremath{\mathbb{R}}}
\newtheorem{definition}{Definition}
\newtheorem{theorem}{Theorem}
\newcommand{\dist}{\ensuremath{{\textit dist}}}
\newcommand{\MaxDist}{\ensuremath{\textit{MaxDist}}}

\newcommand{\MinDist}{\ensuremath{\textit{MinDist}}}

\newtheorem{example}{Example}

\title{Complete and Sufficient Spatial Domination of Multidimensional Rectangles}

\author{
 Tobias Emrich \\
  Harman International \\
  \texttt{Tobias.Emrich@harman.com} \\
   \And
 Hans-Peter Kriegel\\
  Ludwig Maximilian University of Munich \\
  \texttt{kriegel@dbs.ifi.lmu.de} \\
   \And
 Andreas Z\"ufle \\
  George Mason University\\
  \texttt{azufle@gmu.edu} \\
   \And
 Peer Kr\"oger\\
  Ludwig Maximilian University of Munich\\
  \texttt{kroeger@dbs.ifi.lmu.de} \\
   \And
 Matthias Renz\\
  Christian-Albrechts-Universität zu Kiel\\
  \texttt{mr@informatik.uni-kiel.de} \\
  }

\begin{document}
\maketitle

\begin{abstract}
Rectangles are used to approximate objects, or sets of objects, in a plethora of applications, systems and index structures. Many tasks, such as nearest neighbor search and similarity ranking, require to decide if objects in one rectangle $A$ may/must/must not be closer to objects in a second rectangle $B$, than objects in a third rectangle $R$. To decide this relation of ``Spatial Domination'' it can be shown that using minimum and maximum distances it is often impossible to detect spatial domination. This spatial gem provides a necessary and sufficient decision criterion for spatial domination that can be computed efficiently even in higher dimensional space. In addition, this spatial gem provides an example, pseudocode and an implementation in Python.
\end{abstract}

% keywords can be removed

\section{Introduction}
Minimal bounding rectangles (MBRs) are used as
object approximations in a plethora of different
applications. For example, MBRs are used to bound spatially extended objects and MBRs are used as spatial key for spatial access methods such as the R-Tree \cite{Gut84,BecKriSchSee90}. MBR
approximations have also become very popular for uncertain
data\-bases \cite{BesSolIly08,ChenCheng07,CCMC08,Ckp04,LianChen09}
 to approximate all possible locations of an uncertain object.
 
\begin{wrapfigure}{r}{0.5\textwidth}  
    \centering
        \includegraphics[width=0.4\columnwidth]{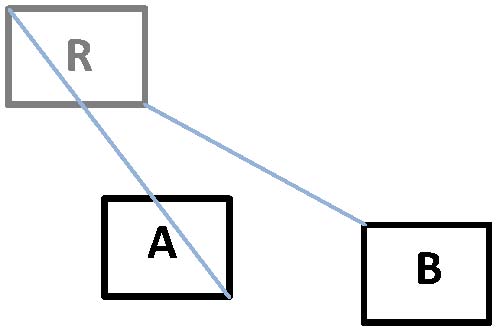}
    \caption{Spatial pruning on MBRs.}
    \label{fig:basic}
\end{wrapfigure}

Rectangular approximations are commonly integrated into spatial query processing as a filter-step, to identify true hits and true drops efficiently based on rectangular approximations only. 
Most types of spatial/similarity queries, including $k$-nearest neighbor ($k$NN).
queries, reverse $k$-nearest neighbor (R$k$NN) queries, and
ranking queries, commonly require the following information. Given
three rectangles $A$, $B$, and $R$ in a multi-dimensional space
$\RR^d$, the task is to determine whether
object $A$ is definitely closer to $R$ than $B$ w.r.t. a distance
function defined on the objects in $\RR^d$. If this is the case,
we say $A$ \emph{dominates} $B$ w.r.t.\ $R$. An example of such a
situation is depicted in Figure \ref{fig:basic}. This concept of domination is a
central problem to identify true hits and true drops
(pruning). For example, in case of a $1$NN query around $R$, we
can prune $B$ if it is dominated by $A$ w.r.t.\ $R$ and for an
R$1$NN query around $R$, we can prune $B$ if $A$ dominates $R$
w.r.t.\ $B$.

The domination problem is trivial for point objects. However,
applied to rectangles, the domination problem is much more
difficult to solve. Traditionally, the minimal distance and
maximal distance between rectangles (min/max-dist) are used to decide which
object is closer to another object: If the maximum distance between $R$ and $A$ is lower than the minimum distance between $R$ and $B$, then, for any points $r\in R$, $a\in A$ and $b\in B$, it must hold that $a$ is closer to $r$ than $b$. While this implication is correct, the backward direction does not hold. Thus, min/max-dist provides a sufficient but not a complete decision criterion. To illustrate this problem, consider the example shown in Figure \ref{fig:minmax}. In this example, we can guarantee that point $a$ spatially dominates point $b$ with respect to $R$. That is, any point in $R$ must be closer to point $a$ than to point $b$, as any point in $R$ is located on the same side of the equi-distance line $H_{ab}$ (or Voronoi-line) between $a$ and $b$. Yet, in this example, we still have $MaxDist(a,R)>MinDist(b,R)$, thus the min/max-dist approach does not allow to identify this spatial domination. The error of min/max-dist is incurred by using two different locations of $R$ for the computation of mindist and maxdist. 

Another existing method for detecting spatial domination on rectangle has been proposed in~\cite{emrich2009incremental} which exploits convexness of rectangles to check only the combination of corners of all three rectangles $A$, $B$, and $R$. The drawback of this approach is the high run-time, which scale exponentially in the dimensionality of the data space.

This spatial gem revisits a necessary and sufficient decision criterion for spatial domination of rectangles \cite{emrich2010boosting} that allows to efficiently scale to high dimensionality. We summarize the theory, describe application, illustrate examples and provides an implementation in Python.

\begin{figure}
    \centering
        \includegraphics[width=0.75\columnwidth]{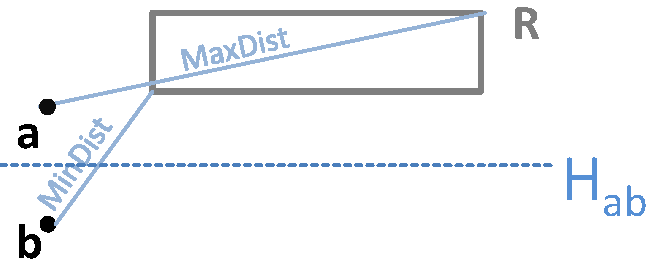}
    \caption{MBR pruning example: Incompleteness of Min/Max}
    \label{fig:minmax}
\end{figure}

\section{Necessary and Sufficient Domination of Rectangles}\label{sec:dom}
The problem of spatial domination is formally defined as follows.

\begin{definition}[Domination]\label{def:domination}
Let $A,B,R\subseteq\RR^d$ be rectangles and $\dist$ be an $L_P$ norm to measure distance between points in $\RR^d$ such as Euclidean distance ($P=2$) and Manhattan distance ($P=1$). The rectangle $A$
dominates $B$ w.r.t.\ $R$ iff for all points $r\in R$ it holds that every point
$a \in A$ is closer to $r$ than any point $b \in B$, i.e.
\begin{equation}\label{eq:def}
Dom_R(A,B) \Leftrightarrow \forall r \in R, \forall a \in A, \forall b \in B: \dist(a, r) < \dist(b, r)
\end{equation}
\end{definition}
Equation \ref{eq:def} enumerates an (uncountably) infinite number of triples of possible point locations and thus, cannot be computed directly in this form. 
An decision criterion that is complete, sufficient, and can be computed in $O(d)$ time has been proposed in \cite{emrich2010boosting} as follows:

\begin{theorem}[Complete and Sufficient Domination]\label{the:decision}
Let $A,B,R\subseteq\RR^d$ be rectangles and $\dist$ be an $L_P$ norm to measure distance between points in $\RR^d$. Further, let $\MinDist(X,x)$ and $\MaxDist(X,x)$ denote the minimum and maximum distance between a (one-dimenstional) interval $X=[X^{min},X^{max}]$ and a scalar $x$, respectively:

\begin{minipage}{.48\linewidth}
$$
MinDist(X,x) = \left\{\begin{array}{ll}
        0, & \mbox{for } x\in X\\
        X^{min}-x, & \mbox{for } x<X^{min}\\
        x-X^{max}, & \mbox{for } x>X^{max}
        \end{array}\right\}
$$
\end{minipage}
\begin{minipage}{.48\linewidth}
$$
MaxDist(X,x) = \max(|x-X^{min}|,|x-X^{max}|)
$$
\end{minipage}
Then, the following equivalence holds:
\begin{equation}\label{eq:criterion}
Dom_R(A,B) \Leftrightarrow \sum_{i=1}^d \max\limits_{r\in \{R_i^{min}, R_i^{max}\}}(\MaxDist
(A_i,r)^P - \MinDist (B_i,r)^P) < 0
\end{equation}
\end{theorem}

\begin{proof}
A detailed inference of Theorem \ref{the:decision} can be found in \cite{emrich2010boosting}.
\end{proof}
\clearpage

\begin{example}
Let us revisit the example in Figure \ref{fig:minmax} using Euclidean distance ($P=2$), and assume the following coordinates of points and rectangles in this example: $a=(0,2)$, $b=(0,0)$, $R_{min}=(2,2)$, and $R_{max}=(10,4)$.
Using Equation \ref{eq:criterion} we obtain, for the first dimensions ($i=1$): 
$$
\max\limits_{r_1\in \{R_1^{min}, R_1^{max}\}}(\MaxDist
(A_1,r_1)^2 - \MinDist (B_1,r_1)^2)=
$$
$$
\max\limits_{r_i\in \{2,10\}}(\MaxDist
(A_1,r_1)^2 - \MinDist (B_1,r_1)^2)=
$$
$$
\max(\MaxDist
([0,0],2)^2 - \MinDist ([0,0],2)^2)), \MaxDist
([0,0],10)^2 - \MinDist ([0,0],10)^2)=
$$
$$
\max(4-4,100-100)=max(0,0)=0
$$
and for the second dimension we get:
$$
\max\limits_{r_2\in \{R_2^{min}, R_2^{max}\}}(\MaxDist
(A_2,r_2)^2 - \MinDist (B_2,r_2)^2)=
$$
$$
\max\limits_{r_i\in \{2,4\}}(\MaxDist
(A_2,r_2)^2 - \MinDist (B_2,r_2)^2)=
$$
$$
\max(\MaxDist
([2,2],2)^2 - \MinDist ([0,0],2)^2)), \MaxDist
([2,2],10)^2 - \MinDist ([0,0],10)^2)=
$$
$$
\max(0-4,64-100)=max(-4,-36)=-4
$$
The sum over the two dimensions yields $0-4=-4$, and since $-4<0$, the inequation of Equation \ref{eq:criterion} is satisfied, and thus, $Dom_R(a,b)$ holds.
\end{example}

\section{Implementation}
This section provides an implementation for the complete and sufficient spatial domination decision criterion sketched in Section \ref{sec:dom} and described in detail in \cite{emrich2010boosting}. Pseudocode can be found in Algorithm \ref{alg:Pseudo}, which takes three $d$-dimensional rectangle and an $L_P$ norm as input, and decides if $A$ spatially dominates $B$ with respect to $R$. The algorithm iterates of all dimensions in Line~2. For each dimension, the algorithm uses the minimum and maximum point of rectangle $R$, and compares minDist and maxDist of these points to rectangles $A$ and $B$ in lines~4-5. Algorithms to compute maxDist and minDist are shown in Algorithm~\ref{alg:maxdist} and Algorithm~\ref{alg:mindist}, respectively. 
The larger of the two is aggregated into the final result in lines~6-10.  

For an implementation of Algorithm~\ref{alg:Pseudo} in Python, see \url{https://github.com/azufle/Spatial_Gems_Spatial_Domination}.

\begin{algorithm}
\SetKwInOut{Input}{input}
\SetKwInOut{Output}{output}
\DontPrintSemicolon
\Input{$P$ \tcp*{Employed $L_P$ norm to measure distance}}
\Input{$d$ \tcp*{Number of dimensions}}
\Input{$A,B,R$ \tcp*{Three Rectangles in $R^d$.}}
\Output{A boolean predicate to decide if $A$ spatially dominates $B$ w.r.t. $R$}
\BlankLine
$sum \leftarrow 0$ \tcp*{Initialize the sum}
\For{$i \in 1:d$\tcp*{Iterate over each dimension}} { 
    $max1 \leftarrow MaxDist(A_i,R^{min}_i)^P-MinDist(B_i,R^{min}_i)^P$
    \tcp*{First argument of the maximum}
    $max2 \leftarrow MaxDist(A_i,R^{max}_i)^P-MinDist(B_i,R^{max}_i)^P$
    \tcp*{Second argument of the maximum}
    \uIf{$max1\geq max2$ \tcp*{Add the larger argument to the sum}} 
    {
        $sum\leftarrow sum+max1$
    }
    \uElse{$sum\leftarrow sum+max2$}
}
\Return $sum<0$ \tcp*{Return True if $sum<0$ and False otherwise}
\caption{Complete and Sufficient Domination of MBRs}
\label{alg:Pseudo}
\end{algorithm}

\begin{algorithm}[t]
\SetKwInOut{Input}{input}
\SetKwInOut{Output}{output}
\DontPrintSemicolon
\Input{$X=[X^{min},X^{max}]$ \tcp*{An interval from $X^{min}$ to $X^{max}$}}
\Input{$x$ \tcp*{A scalar in $\RR$}}
\Output{The maximum distance between $X$ and $x$ (for any point in $X$)}
\BlankLine
\uIf{$|x-X^{min}|\geq |x-X^{max}|$}
{
  \Return $|x-X^{min}|$
}
\uElse{
   \Return $|x-X^{max}|$
}
\caption{Maximum Distance between an Interval and a Scalar}
\label{alg:maxdist}
\end{algorithm}

\begin{algorithm}[t]
\SetKwInOut{Input}{input}
\SetKwInOut{Output}{output}
\DontPrintSemicolon
\Input{$X=[X^{min},X^{max}]$ \tcp*{An interval from $X^{min}$ to $X^{max}$}}
\Input{$x$ \tcp*{A scalar in $\RR$}}
\Output{The minimum distance between $X$ and $x$ (for any point in $X$)}
\BlankLine
\uIf{$x<X^{min}$}
{
  \Return $X^{min}-x$
}
\uElseIf{$x\leq X^{max}$ }
{
  \Return $0$ \tcp*{Case where $x\in X$}
}
\uElse{
   \Return $x-X^{max}$ \tcp*{Case where $x> X^{max}$}
}
\caption{Minimum Distance between an Interval and a Scalar}
\label{alg:mindist}
\end{algorithm}

\newpage
\section{Scenarios and Applications}
\label{sec:applications}
In this section, we will show how the concepts of spatial domination can be used to accelerate candidate pruning used in
similarity search and recommendation systems.

%\begin{figure}
%    \centering
%        \includegraphics[width=0.5\columnwidth]{single-hyperplane-pruning.jpg}
%    \caption{Scenario 1: Single hyperplane based spatial pruning}
%    \label{fig:single-hyperplane-prunin}
%\end{figure}

\subsection{Reverse Nearest Neighbor Search}
\label{subsec:Applications_RNN}
The reverse k-nearest neighbor query problem is given a point q, retrieve all the data points that have q as one of their k nearest neighbors. The geometrical pruning based solution of this problem introduced in \cite{tao2004reverse} overcomes the problem of other RkNN approaches including 1) support arbitrary values of k, 2) can efficiently deal with database updates, and 3) are applicable to arbitrary-dimensional feature spaces. The basic operation of determining if an object or a page region of a spatial index, e.g. an R-tree, can be pruned based on a bisecting hyperplane defined by two multi-dimensional points can be efficiently solved with our spatial domination solution.    

\subsection{Multi-Preference Recommendation}
Multi-preference recommendation and multi-criteria decision making based on top-k query processing has become a hot topic and has been studied extensively in the last two decades. In the context of multi criteria top-k query processing, computational geometry driven top-k query processing models have gained a lot of interest in recent years
\cite{mouratidis2019geometric,mouratidis2018exact,tang2017determining,yang2018querying,qian2015learning} and could be successfully adapted to several variants of top-k queries in multi-criteria settings
\cite{mouratidis2018exact}. The principal idea of the above mentioned line of work is to translate object domination according to multi-criteria preferences into hyperplane bisection of the multidimensional preference space.
A key operation shared by all these approaches is, given a halfspace $S$ defined by hyper-plane bi-section of the multidimensional preference space, and an axis parallel rectangle $R$, the determination of which of the following cases is true: case 1) $R$ is completely covered by $S$, case 2) $R$ intersects $S$, or case 3) $R$ does not intersect $S$ at all, as illustrated in Figure \ref{fig:r-domination}.  

\begin{figure}
    \centering
        \includegraphics[width=0.7\columnwidth]{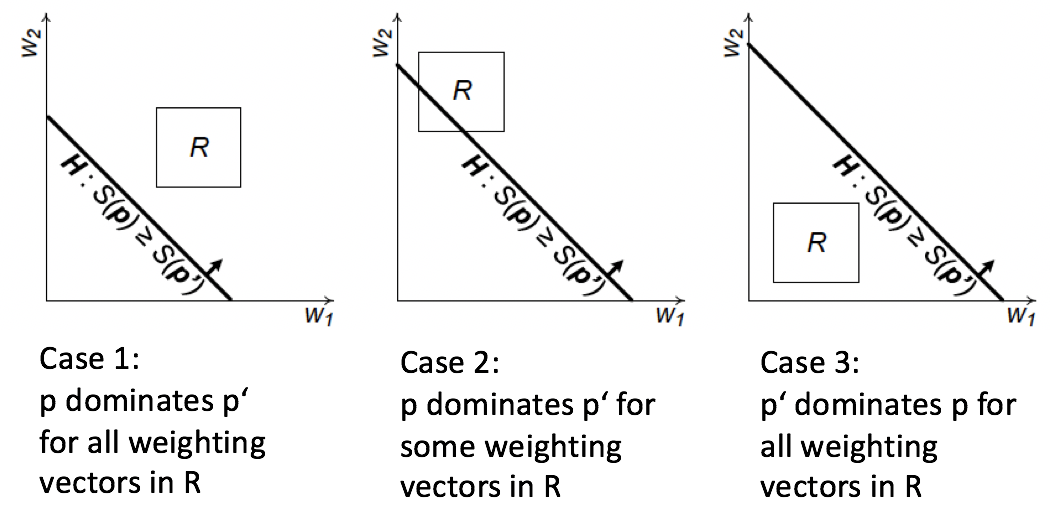}
    \caption{R-domination identification based on hyperplane bisection \cite{mouratidis2018exact}}
    \label{fig:r-domination}
\end{figure}

%A 3) Continuous Nearest Neighbor Query Processing / NN Query for moving objects

%Scenario 2: Bisecting hypercurve defined by two rectangles
%\begin{figure}
%    \centering
%        \includegraphics[width=0.5\columnwidth]{multi-hyperplane-pruning.jpg}
%    \caption{Scenario 2: Multi-hyperplane based spatial pruning}
%    \label{fig:mulit-hyperplane-pruning}
%\end{figure}

\bibliographystyle{unsrt}  
\bibliography{literature,references}  %%% Remove comment to use the external .bib file (using bibtex).

\begin{thebibliography}{10}

\bibitem{Gut84}
A.~Guttman.
\newblock {R-Trees}: A dynamic index structure for spatial searching.
\newblock In {\em ACM SIGMOD}, pages 47--57, 1984.

\bibitem{BecKriSchSee90}
N.~Beckmann, H.-P. Kriegel, R.~Schneider, and B.~Seeger.
\newblock {The R*-Tree}: An efficient and robust access method for points and
  rectangles.
\newblock In {\em ACM SIGMOD}, 1990.

\bibitem{BesSolIly08}
George Beskales, Mohamed~A. Soliman, and Ihab~F. IIyas.
\newblock Efficient search for the top-k probable nearest neighbors in
  uncertain databases.
\newblock {\em Proc. VLDB Endow.}, 1(1):326--339, 2008.

\bibitem{ChenCheng07}
J.~Chen and R.~Cheng.
\newblock Efficient evaluation of imprecise location-dependent queries.
\newblock In {\em ICDE}, 2007.

\bibitem{CCMC08}
R.~Cheng, J.~Chen, M.~Mokbel, and C.~Chow.
\newblock Probabilistic verifiers: Evaluating constrained nearest-neighbor
  queries over uncertain data.
\newblock In {\em ICDE}, 2008.

\bibitem{Ckp04}
R.~Cheng, D.~Kalashnikov, and S.~Prabhakar.
\newblock Querying imprecise data in moving object environments.
\newblock In {\em TKDE}, 2004.

\bibitem{LianChen09}
X.~Lian and L.~Chen.
\newblock Probabilistic inverse ranking queries over uncertain data.
\newblock In {\em DASFAA}, 2009.

\bibitem{emrich2009incremental}
Tobias Emrich, Hans-Peter Kriegel, Peer Kr{\"o}ger, Matthias Renz, and Andreas
  Z{\"u}fle.
\newblock Incremental reverse nearest neighbor ranking in vector spaces.
\newblock In {\em International Symposium on Spatial and Temporal Databases},
  pages 265--282. Springer, 2009.

\bibitem{emrich2010boosting}
Tobias Emrich, Hans-Peter Kriegel, Peer Kr{\"o}ger, Matthias Renz, and Andreas
  Z{\"u}fle.
\newblock Boosting spatial pruning: on optimal pruning of mbrs.
\newblock In {\em Proceedings of the 2010 ACM SIGMOD International Conference
  on Management of data}, pages 39--50. ACM, 2010.

\bibitem{tao2004reverse}
Yufei Tao, Dimitris Papadias, and Xiang Lian.
\newblock Reverse knn search in arbitrary dimensionality.
\newblock In {\em Proceedings of the Thirtieth international conference on Very
  large data bases-Volume 30}, pages 744--755. VLDB Endowment, 2004.

\bibitem{mouratidis2019geometric}
Kyriakos Mouratidis.
\newblock Geometric top-k processing: Updates since mdm'16 [advanced seminar].
\newblock In {\em 2019 20th IEEE International Conference on Mobile Data
  Management (MDM)}, pages 1--3. IEEE, 2019.

\bibitem{mouratidis2018exact}
Kyriakos Mouratidis and Bo~Tang.
\newblock Exact processing of uncertain top-k queries in multi-criteria
  settings.
\newblock {\em Proceedings of the VLDB Endowment}, 11(8):866--879, 2018.

\bibitem{tang2017determining}
Bo~Tang, Kyriakos Mouratidis, and Man~Lung Yiu.
\newblock Determining the impact regions of competing options in preference
  space.
\newblock In {\em Proceedings of the 2017 ACM International Conference on
  Management of Data}, pages 805--820. ACM, 2017.

\bibitem{yang2018querying}
Guolei Yang and Ying Cai.
\newblock Querying a collection of continuous functions.
\newblock {\em IEEE Transactions on Knowledge and Data Engineering},
  30(9):1783--1795, 2018.

\bibitem{qian2015learning}
Li~Qian, Jinyang Gao, and HV~Jagadish.
\newblock Learning user preferences by adaptive pairwise comparison.
\newblock {\em Proceedings of the VLDB Endowment}, 8(11):1322--1333, 2015.

\end{thebibliography}

%\bibitem{kour2014real}
%George Kour and Raid Saabne.
%\newblock Real-time segmentation of %on-line handwritten arabic script.
%\newblock In {\em Frontiers in Handwriting Recognition (ICFHR), 2014 14th
%  International Conference on}, pages 417--422. IEEE, 2014.
%
%\bibitem{kour2014fast}
%George Kour and Raid Saabne.
%\newblock Fast classification of handwritten on-line arabic characters.
%\newblock In {\em Soft Computing and Pattern Recognition (SoCPaR), 2014 6th
%  International Conference of}, pages 312--318. IEEE, 2014.

%\bibitem{hadash2018estimate}
%Guy Hadash, Einat Kermany, Boaz Carmeli, Ofer Lavi, George Kour, and Alon
%  Jacovi.
%\newblock Estimate and replace: A novel approach to integrating deep neural
%  networks with existing applications.
%\newblock {\em arXiv preprint arXiv:1804.09028}, 2018.

%\end{thebibliography}

\end{document}